\documentclass[11pt]{article}

\usepackage[margin=1in]{geometry}

\usepackage{times}

\usepackage{soul}
\usepackage{url}
\usepackage[hidelinks]{hyperref}
\usepackage[utf8]{inputenc}
\usepackage[small]{caption}
\usepackage{graphicx}
\usepackage{amsmath}
\usepackage{amsfonts}
\usepackage{booktabs}
\usepackage{color}
\usepackage{xfrac}

\usepackage[ruled,linesnumbered,vlined]{algorithm2e}

\newenvironment{proof}{{\bf Proof:  }}{\hfill\rule{2mm}{2mm}\vspace*{5pt}}
\newenvironment{proofof}[1]{{\vspace*{5pt} \noindent\bf Proof of #1:  }}{\hfill\rule{2mm}{2mm}\vspace*{5pt}}

\newtheorem{claim}{Claim}[section]
\newtheorem{theorem}{Theorem}[section]

\newtheorem{lemma}[theorem]{Lemma}

\newtheorem{example}[theorem]{Example}
\newtheorem{definition}[theorem]{Definition}

\title{Budget-feasible Maximum Nash Social Welfare Allocation is Almost Envy-free}

\author{Xiaowei Wu$^1$ \hspace{30pt} Bo Li$^2$ \hspace{30pt} Jiarui Gan$^3$\\
$^1$IOTSC, University of Macau\\\
$^2$Department of Computing, The Hong Kong Polytechnic University\\
$^3$Max Planck Institute for Software Systems\\
\texttt{xiaoweiwu@um.edu.mo, comp-bo.li@polyu.edu.hk, jrgan@mpi-sws.org}}

\date{}

\begin{document}

\maketitle

\begin{abstract}
The {\em Nash social welfare} (NSW) is a well-known social welfare measurement that balances individual utilities and the overall efficiency.
In the context of fair allocation of indivisible goods, it has been shown by Caragiannis et al. (EC 2016 and TEAC 2019) that an allocation maximizing the NSW is envy-free up to one good (EF1).
In this paper, we are interested in the fairness of the NSW in a budget-feasible allocation problem, in which each item has a cost that will be incurred to the agent it is allocated to, and each agent has a budget constraint on the total cost of items she receives.
We show that a budget-feasible allocation that maximizes the NSW achieves a $1/4$-approximation of EF1 and the approximation ratio is tight.
The approximation ratio improves gracefully when the items have small costs compared with the agents' budgets;
it converges to $1/2$ when the budget-cost ratio approaches infinity.
\end{abstract}

\section{Introduction}

Fairness and efficiency are two of the primary considerations in public economics.
Yet, these two objectives are often unaligned and even conflict with each other. 
The {\em Nash social welfare} (NSW) --- proposed by Nash as a solution for bargaining
problems \cite{nash1950bargaining,kaneko1979nash} --- is a well-known social welfare function that balances individual utilities and the overall efficiency.
In more details, the NSW of a solution $\mathbf{x}$ is defined as the product of the agents' values $v_i(\mathbf{x})$ provided by this solution.
As noted in \cite{conf/sigecom/BranzeiGM17}, it can be viewed as a special case of a family of functions known as {\em generalized (power) means}:
\[
\textstyle M_p(\mathbf{x}) =  \left(\frac{1}{n}\sum_{i} v_i(\mathbf{x})^p \right)^{1/p}.
\]
When $p = 1$, $M_p(\mathbf{x})$ defines the utilitarian social welfare, i.e., the average of all agents' values; and when $p \to -\infty$, $M_p(\mathbf{x})$ defines the egalitarian social welfare, i.e., the minimum value received by an agent.
The NSW corresponds to exactly the limit of $M_p(\mathbf{x})$ when $p \to 0$, i.e., $\left(\prod_{i \in N} v_i(\mathbf{x}) \right)^{1/n}$.
Thus, it serves a middle ground between the two other cases. A Max-NSW allocation, i.e., an allocation that maximizes the NSW, naturally leads to balanced values among the agents while it also nudges the overall efficiency towards the maximum.

In the context of fair allocation of indivisible goods,
it has been shown by Caragiannis et al.~\cite{caragiannis2016unreasonable,caragiannis2019unreasonable}
that a Max-NSW allocation is {\em envy-free up to one good} (EF1); 
namely, it ensures that every agent does not envy the bundle of any other agent by more than one item. The fairness of the NSW has since attracted much attention from the fair allocation community.

\subsection{Our Results}

In this paper, we study the fairness of the NSW in a budget-feasible allocation problem.
In this problem, every item $j$ has a cost of $c_j$ that will be incurred to the agent it is allocated to, and every agent $i$ has a budget $B_i > 0$ that can be used to pay for the additive cost of the bundle $S$ of items she receives, i.e., $c(S) = \sum_{j \in S} c_j\leq B_i$.
An allocation $(X_1, \cdots, X_n)$ is a collection of disjoint subsets of items, whereby $X_i$ is the bundle of items allocated to agent $i$, yielding value $v_i(X_i)$ for this agent.
An allocation is {\em budget-feasible} if $c(X_i) \le B_i$ for all $i$, i.e., every agent can afford the cost incurred by this allocation. 
Accordingly, a Max-NSW feasible allocation is equivalent to the solution to the following problem. 
\begin{align*}
\max_{(X_1, \cdots, X_n)} &\quad \prod_{i \in N} v_i(X_i) \\
\text{subject to} &\quad (X_1, \cdots, X_n) \text{ \em is budget-feasible.}
\end{align*}
Note that since allocations are subject to budget constraints, it is not always possible to allocate all items. The unallocated items are assumed to be allocated to a {\em charity} who has unlimited budget.
Indeed, it is not hard to see that a Max-NSW allocation is Pareto-optimal (PO), i.e., we cannot hope to further improve some agent's utility without hurting the other agents.
This immediately implies that no agent has any incentive to exchange her bundle with any subset of unallocated items that fits within her budget.
An interesting question is whether the agents would wish to exchange their bundles with the other agents.

To answer this question, we adapt the classic EF (envy-free) notion to the budget-feasible setting. In our EF notion, an agent $i$ envies another agent $j$ if there is a subset of items allocated to agent $j$ that costs at most $B_i$ and gives agent $i$ a strictly higher value. 
The EF1 notion can be defined similarly: agent $i$ envies agent $j$ for more than one item if there is a subset of $X_j$ that costs at most $B_i$ and gives $i$ a strictly higher value even after removing the most valuable item from this subset. 
An allocation is said to be EF (resp., EF1) if no agent envies any other agent (resp., for more than one item).
We show that, unlike the setting without budget constraints, a Max-NSW allocation is no longer always EF1. 
However, it remains a good approximation of EF1.
Our main contributions are summarized by the following two results, where we call an allocation $\alpha$-EF1 is it achieves an $\alpha$ approximation of EF1, for $\alpha\in [0,1]$. 

\medskip
\noindent{\bf Main Result 1.} {\em A Max-NSW allocation is $\frac{1}{4}$-EF1 and PO.}

\medskip
\noindent
We also show that the approximation ratio we proved cannot be improved.
Specifically, we construct an instance for which a Max-NSW allocation is exactly $ 1/4$-EF1.
We observe that this instance requires some items to have very high costs compared with the agents' budgets, i.e., the budget-cost ratio is $\kappa = \min_{i,j}\frac{B_i}{c_j} = 1$.
This motivates our study of the {\em large budget case} where $\kappa$ is large.
Interestingly, we find that the approximation guarantee improves gracefully with $\kappa$ and eventually converges to ${1}/{2}$ when $\kappa$ goes to infinity.

\medskip
\noindent{\bf Main Result 2.} {\em A Max-NSW allocation is $\left( \frac{1}{2}-\frac{5}{\kappa^{1/4}} \right)$-EF1.}

\medskip
\noindent
Note that instances with large budget are also a typical setting of the bin packing problem that has received considerable interests in the literature~\cite{esa/FeldmanHKMS10,ec/DevanurJSW11,icalp/MolinaroR12,stoc/KesselheimTRV14}.
We remark that when $\kappa$ goes to infinity our model does not degenerate to the one without budget constraints, because the number of items can be infinity as well. This is why the approximation ratio does not converge to $1$ as in the setting without budget constraints.

\paragraph{Technical Novelty.}
We remark that deriving the above results requires techniques very different from those employed in previous work. The approach of Caragiannis et al.~\cite{caragiannis2016unreasonable,caragiannis2019unreasonable} for proving the EF1-ness of Max-NSW allocation is by contradiction: if a Max-NSW allocation were not EF1, then by reallocating some items, the NSW could be improved.
In our setting, such reallocations need to conform to the budget constraints.
We may need to discard some items when reallocating items to an agent.
Thus, applying the same argument requires a larger difference between the agents' values.
Indeed, this is why the Max-NSW allocation is only approximately EF1 in our setting, instead of being exactly EF1.

Our analysis for the second main result is more challenging.
We observe that if the items are divisible, then any Max-NSW allocation must be $1/2$-EF1.
To convert the analysis to the indivisible case, we show that in the divisible analysis at most a constant number of items are fractionally allocated.
Then it suffices to give a rounding scheme for a small number of fractional items.
A main challenge arises when we have to deal with fractional items with large values. We resolve this issue by using a fine-tuned classification of items into heavy and light ones and applying different analyses depending on the class the majority of the items belong to.

\subsection{Related Work}
There is a large literature on fair division and NSW.
Due to space limit, we only discuss the most related previous work.
The pioneering work of Caragiannis et al.~\cite{caragiannis2016unreasonable,caragiannis2019unreasonable} triggered 
a series of follow-up interests in the fairness of the NSW.
It is shown in \cite{biswas2018fair} that with matriod constraints a Max-NSW allocation is always EF1 if the agents have identical valuations.
More recently, Amanatidis et al.~\cite{conf/ijcai/AmanatidisBFHV20} proved that a Max-NSW allocation guarantees envy-freeness up to {\em any} item (EFX) if there are at most two possible values for the goods.
The fairness of NSW has also been studied in the setting when the goods are public \cite{conitzer2017fair,fain2018fair}.
Furthermore, the algorithm in \cite{barman2018finding} showed that there is an allocation that is simultaneously $1/1.45$-approximate Max-NSW, Pareto optimal, and $(1-\epsilon)$-approximate EF1. Similarly, the algorithm in \cite{garg2019improving} computes an allocation that is $1/2$-approximate Max-NSW, proportional up to one item, $1/(2n)$-approximate maximin share fair,
and Pareto optimal and can be found in polynomial time. 
No previous work has considered the setting with budget constraints to the best of our knowledge.

Our work is conceptually related to a line of work on fair division that involves a charity in the models.
In our budget-feasible setting, unallocated items are assumed to be donated to a charity.
The concept of charity has also been used to simplify analysis of EFX allocations as was in \cite{caragiannis2016unreasonable,caragiannis2019unreasonable}.
Caragiannis et al. also showed that there is an EFX allocation on a subset of items that is a $1/2$-approximation of Max-NSW over the original set of items \cite{caragiannis2019envy}.
Similarly, Chaudhury et al. \cite{chaudhury2020little} proved that by donating no more than $n$ items, 
there is an EFX allocation for the remaining items and no agent envies the charity.

The complexity of computing a Max-NSW allocation has been studied for various types of valuation functions \cite{ramezani2009nash,darmann2015maximizing}.
The problem is APX-hard for additive valuation functions \cite{nguyen2014computational,lee2017apx}. 
There has been some effort towards designing efficient algorithms to compute approximate Max-NSW allocations \cite{nguyen2014minimizing,cole2015approximating,cole2017convex}.
The state-of-the-art approximation ratio of NSW is $1.45$ by \cite{barman2018finding}. 

\section{Preliminaries}

In our model, a set $M$ of $m$ goods will be allocated to a set $N$ of $n$ agents.
Every item $j \in M$ has a cost $c_j\geq 0$ and every agent $i \in N$ has a budget $B_i\geq 0$.
Each agent $i$ has an additive valuation function $v_i : 2^M \rightarrow \mathbb{R}^+\cup\{0\}$ on the subsets of items.
Equivalently, we can express  the valuation function $v_i$ as a vector $v_i = (v_{i1}, \cdots, v_{im})$, where $v_{ij} \in \mathbb{R}^+\cup\{0\}$ denotes agent $i$'s value for item $j$.
Upon receiving a subset $S \subseteq M$ of items, an agent $i$ obtains an additive value $v_i(S) = \sum_{j \in S} v_{ij}$; in addition, the allocation incurs an additive cost $c(S) = \sum_{j \in S} c_j$ to agent $i$, which needs to be paid using her budget. 
We assume that the values for the items and the cost of the allocation are not mutually convertible. For example, the cost may represent the size of the items and the budget represents the capacity of her knapsack, which cannot be converted from the items' values.

An allocation $\mathbf{X} = (X_0,X_1, \cdots, X_n)$ is a partition of the item set $M$ into $n+1$ bundles, where $X_i \subseteq M$ is the bundle of items allocated to agent $i$ for each $i\in N$ and $X_0$ is the set of unallocated item (recall that due to the budget constraints, we cannot guarantee $\bigcup_{i=1}^n X_i = M$).
The unallocated items can be thought of as a donation to a special agent $0$, i.e., a {\em charity}, who has an unbounded budget and value $0$ for every item (in the sense that she does not envy any other agents no matter what they get).
An allocation $\mathbf{X}$ is \emph{budget-feasible} if $c(X_i) \leq B_i$, for all $i\in N$.

In the remainder of this paper, we require every allocation to be feasible.
When we write ``an allocation'' we mean ``a feasible allocation'' unless otherwise stated.
For convenience, for a set $X$ and an element $j$, we will write $X\cup \{j\}$ as $X+j$ and $X\setminus \{j\}$ as $X-j$ throughout the paper.

\begin{definition}[Max-NSW]
An allocation $\mathbf{X^*}$ is a Max-NSW allocation if it maximizes the Nash social welfare, i.e., $\prod_{i\in N} v_i(X_i^*) \ge \prod_{i\in N} v_i(X_i)$ for any allocation $\mathbf{X}$.
\end{definition}

In the budget-feasible setting, an agent cannot hope to take a bundle that costs more than her budget. We adapt the EF notion accordingly and define it as follows.

\begin{definition}[EF]
For any $\alpha\in [0,1]$, an allocation $\mathbf{X}$ is {\em $\alpha$-approximate envy-free ($\alpha$-EF)} if for any two agents $i, j \in N$ and any $S \subseteq X_j$ with $c(S) \le B_i$, it holds that
	\begin{equation*}
	v_i(X_i) \ge \alpha\cdot v_i(S).
	\end{equation*}
When $\alpha = 1$, $\mathbf{X}$ is also said to be EF. 
\end{definition}
This definition is consistent with the standard EF notion when every agent's budget is $c(M)$.
Since the items are indivisible, (approximate) EF is hard to achieve.
One of the most widely studied relaxed notion is envy-free up to one item.

\begin{definition}[EF1] \label{def:ef1}
	For any $\alpha\in [0,1]$, an allocation $\mathbf{X}$ is {\em $\alpha$-approximate envy-free up to one item ($\alpha$-EF1)} if for any two agents $i, j \in N$ and any $S \subseteq X_j$ with $c(S) \le B_i$, there exists $j\in S$ such that
	\begin{equation*}
	v_i(X_i) \ge  \alpha \cdot v_i(S - j).
	\end{equation*}
When $\alpha = 1$, $\mathbf{X}$ is also said to be EF1. 
\end{definition}

Trivially, an allocation $\mathbf{X} = (M, \emptyset, \dots, \emptyset)$ (i.e., allocating all items to the charity) is already EF and EF1.
Apparently, such an allocation is not efficient for the agents by any reasonable measurement of efficiency.
Therefore, we consider Pareto optimality as a criterion for efficiency.

\begin{definition}[PO] An allocation $\mathbf{X}$ is {\em Pareto optimal (PO)} if there exists no allocation $\mathbf{X}'$ such that $v_i(X'_i) \ge v_i(X_i)$ for all $i \in N$ and $v_j(X'_j) > v_j(X_j)$ for some $j \in N$.
\end{definition}

Note that in a PO allocation $\mathbf{X}$ no agent $i$ can improve her value by exchanging items with the charity (subject to budget constraint), i.e., for any $S \subseteq X_i \cup X_0$ such that $c(S) \le B_i$, it holds that $v_i(X_i) \ge v_i(S)$.

\section{Max-NSW Allocation is $\frac{1}{4}$-EF1 and PO} \label{sec:1/4-approx}

In this section, we present our first main result, which is stated in the theorem below. 

\begin{theorem}\label{thm:Nash-sw-EF1}
	A Max-NSW allocation $\mathbf{X^*}$ is $\frac{1}{4}$-EF1 and PO.
\end{theorem}

Before presenting the proof, we remark that the ${1}/{4}$ approximation ratio is tight, even when agents have identical valuation and budget.
We present an instance for which a Max-NSW allocation is not $c$-EF1, for any constant $c > {1}/{4}$.

\begin{example}
\label{exp:max-nsw-not-cef1}
Let $\epsilon>0$ be arbitrarily close to $0$.
Consider the instance with two agents who have the identical valuation $v$ and budget $B = 1$.
Let there be $1+\frac{1}{\epsilon}$ items, where $v_1=1+\epsilon$, $c_1=1$, $v_j = 4\epsilon$, and $c_j = \epsilon$ for all $j\in\{2,3,\ldots,1+\frac{1}{\epsilon}\}$.
It can be verified that the allocation $\mathbf{X}^*$ with $X_1^* = \{1\}$ and $X_2^* = \{2,\ldots,1+\frac{1}{\epsilon}\}$ maximizes the NSW with $v(X_1^*)\cdot v(X_2^*) = 4(1+\epsilon)$.
However, $\mathbf{X}^*$ is not $c$-EF1, for any $c > \frac{1}{4}$, because $v(X_1^*) = \frac{1-\epsilon^2}{4}\cdot v(X_2^* - j)$ for all $ j\in X_2^*$.
\begin{table}[htb]
	\begin{center}
		\begin{tabular}{ccccc}
			\toprule
			$j$\quad~	&	1			&	2		&	$\ldots$	&	$1+\frac{1}{\epsilon}$ \\
			\midrule
			$v_j$\quad~	&	$1+\epsilon$	&$4\epsilon$	&	$\ldots$	&	$4\epsilon$\\
			$c_j$\quad~	&	1			&$\epsilon$	&	$\ldots$	&	$\epsilon$\\
			\bottomrule
		\end{tabular}
	\end{center}
\end{table}	
\end{example}

Next we proceed to the proof of Theorem~\ref{thm:Nash-sw-EF1}.

\begin{proofof}{Theorem~\ref{thm:Nash-sw-EF1}}
	It is trivial to see that $\mathbf{X^*}$ is PO because if there is another allocation that increases one agent's value without decreasing any other agent's value, it must also increase the NSW, which contradicts the maximum of NSW under $\mathbf{X^*}$.
	In what follows, we prove that $\mathbf{X^*}$  is $\frac{1}{4}$-EF1 by contradiction.
	Suppose that a Max-NSW allocation $\mathbf{X^*}$ is not $ \frac{1}{4}$-EF1, we show that there exists another allocation $\mathbf{X'}$ with a strictly larger NSW, which is a contradiction with $\mathbf{X^*}$ being a Max-NSW allocation.
	
	By assumption, there exist two agents, say $1$ and $2$, such that agent $1$ finds a set $T \subseteq X_2$ such that $c(T) \le B_1$ and 
	\begin{align}\label{eq:nw:ef1}
	v_1(T-j) > 4\cdot v_1(X_1), \qquad \text{ for all } j\in T.
	\end{align}
	For ease of description, we rename the items in $T$ as $\{1,2, \cdots, t\}$ such that $v_{2j} \le v_{2,j+1}$ for any $1 \le j <t$. Let $T_1 = \{j \in T-t : j \text{ is odd}\}$ and $T_2 = \{j \in T -  t : j \text{ is even}\}$. 
	Note that we have either 
	\begin{equation*}
	v_2(T_1) \le v_2(T_2) \le v_2(T_1+ t)\quad
	\text{ or } \quad
	v_2(T_2) \le v_2(T_1) \le v_2(T_2+ t).
	\end{equation*}
	
	In other words, under the valuation of agent $2$, if we assign item $t$ to the bundle with smaller value, then the resulting bundle has value at least that of the other bundle.		
	
	Now we can construct a new allocation $\mathbf{X}'$, where 
	\begin{itemize}
	\item
	$X'_i = X_i$ for all $i \geq 3$, 
	\item
	$X'_1 = T_l$, where $l \in \arg\max_{\ell \in \{1,2\}} v_1(T_\ell)$, 
	\item 
	$X'_2 = X_2 \setminus T_l$, and 
	\item
	$X'_0 = X_0 \cup X_1$. 
	\end{itemize}
	In other words, we first remove all the items originally allocated to agent $1$.	
	Then we pick the bundle agent $1$ prefers in $\{T_1, T_2\}$, and move items in this bundle from $X_2$ to $X_1$.
	Note that under the valuation of agent $2$, the items agent $2$ loses have total value at most $\frac{1}{2}\cdot v_2(T)$.
	In addition, $\mathbf{X}'$ is a feasible allocation as $c(T_l) \le c(T) \le B_1$ and $c(X_2 \setminus T_l) \le c(X_2) \le B_2$.
	By Equation \eqref{eq:nw:ef1}, we then have
	\begin{equation*}
	v_1(X'_1) \geq \frac{1}{2}\cdot v_1(T-j) > 2\cdot v_1(X_1).
	\end{equation*}
	Moreover,
	\begin{equation*}
	v_2(X'_2) = v_2(X_2) - v_2(T_l) \ge \frac{1}{2}\cdot v_2(X_2).
	\end{equation*}
	Since $v_i(X'_i) = v_i(X_i)$ for all $i \geq 3$, it follows that 
	\begin{equation*}
	\prod_{i \in N} v_i(X'_i) > \prod_{i \in N} v_i(X_i),
	\end{equation*}
	which contradicts the fact that $\mathbf{X^*}$ is Max-NSW. 
\end{proofof}

\paragraph{Remark.} While we assume that the valuation functions are additive, our result holds with {\em sub-additive} valuation functions as well. We provide the details in the Appendix~\ref{appendix:sub-additive}.

\section{Improved Ratios with Large Budgets}

We have shown that in general, the $1/4$ approximation ratio can not be improved.
However, one observation is that in the tight instance we presented (Example~\ref{exp:max-nsw-not-cef1}), the budgets are such that $B_1 = B_2 = 1 = \max_{j\in M}c_j$.
In other words, an agent's budget may be exhausted by allocating only one item to her.
Interestingly, we find that when this is not the case, the approximation ratio can be improved.

\subsection{Warm-up Analysis}

As a warm up, we show that if it takes at least two item to exhaust the budget of any agent, then the approximation ratio can be slightly improved to $\frac{3}{11}\approx 0.273$.

\begin{lemma}
	Suppose $B_i \geq 2\cdot c_j$ for all $i\in N$ and $j\in M$. Then a Max-NSW allocation $\mathbf{X^*}$ is $\frac{3}{11}$-EF1.
\end{lemma}
\begin{proof}
	Suppose for the sake of contradiction that a Max-NSW allocation $\mathbf{X^*}$ is not $\frac{3}{11}$-EF1.
	Without loss of generality, let agents $1$ and $2$ be the two agents that block $\mathbf{X}^*$ from being $\frac{3}{11}$-EF1.
	In other words, there exists $T\subseteq X_2$ with $c(T)\le B_1$ such that for any $j \in T$, we have
	\begin{equation*}
	v_1(T - j) > \frac{11}{3}\cdot v_1(X_1).
	\end{equation*}
	
	We use a similar approach to the proof of Theorem~\ref{thm:Nash-sw-EF1} and partition $T$ into three subsets $T_1$, $T_2$, and $\{t\}$ such that 
	\begin{equation*}
	v_2(T_1) \le v_2(T_2) \le v_2(T_1\cup\{t\}).
	\end{equation*}
	
	\paragraph{Case 1.} If there exists $l\in\{1,2\}$ such that $v_1(T_l) > 2\cdot v_1(X_1)$, we construct a new allocation $\mathbf{X}'$ with $X'_i = X_i$ for all $i\geq 3$, $X'_1 = T_l$, $X'_2 = X_2 \setminus T_l$, and $X'_0 = X_0 \cup X_1$.  By this construction, we have that $v_1(X'_1) > 2\cdot v_1(X_1)$ and $v_2(X'_2) \le \frac{1}{2}\cdot v_2(X_2)$; thus,
	\begin{equation*}
	\prod_{i \in N} v_i(X'_i) > \prod_{i \in N} v_i(X_i),
	\end{equation*}
	which contradicts the fact that $\mathbf{X^*}$ is a Max-NSW allocation.
	
	\paragraph{Case 2.} Otherwise, we have $\max\{ v_1(T_1),v_1(T_2) \} \leq 2\cdot v_1(X_1)$. Since 
	\begin{equation*}
	v_1(T_1) + v_1(T_2) = v_1(T_1 \cup T_2) > \frac{11}{3}\cdot v_1(X_1),
	\end{equation*}
	we have
	\begin{equation*}
	\min\{ v_1(T_1),v_1(T_2) \} > \frac{5}{3}\cdot v_1(X_1).
	\end{equation*}
	Pick an arbitrary $l \in \arg\min_{\ell \in \{1,2\}} c(T_\ell)$.
	Since $c(T_1) + c(T_2) \le B_1$, we have $c(T_l) \le \frac{1}{2}\cdot B_1$.
		
	We construct a new allocation by allocating all the items in $T_l$ to agent $1$, along with a sufficiently valuable portion of items in $X_1$. 
	To ensure that the budget feasibility, we need the following result.
	
	\begin{claim}\label{claim:keep-1/3}
		There exists $\hat{X}_1\subseteq X_1$ such that $c(\hat{X}_1) \le \frac{1}{2}\cdot B_1$ and $v_1(\hat{X}_1) \ge \frac{1}{3}\cdot v_1(X_1)$. 
	\end{claim}
	\begin{proof}
		If $c(X_1) \le \frac{1}{2}\cdot B_1$, we can let $\hat{X}_1 = X_1$. Note that under this case, $v_1(\hat{X}_1) = v_1(X_1)$.
		
		Now suppose $c(X_1) > \frac{1}{2}\cdot B_1$.
		We create a set $T$ and let it be an empty set initially.
		Then we add into $T$ items in $X_1$ one at a time until $c(T) > \frac{1}{2}\cdot B_1$. Let $j$ be the last item added into $T$. We have
		\begin{equation*}
		c(T) > \frac{1}{2}\cdot B_1 \text{ and } c(T-j)\leq \frac{1}{2}\cdot B_1.
		\end{equation*}
		Consider the partition of $X_1$ into three sets: $T-j$, $\{j\}$ and $X_1\setminus T$.
		Note that all three sets have cost at most $\frac{1}{2}\cdot B_1$.
		Let $\hat{X}_1$ be the set with maximum value (under $v_1$) among $T-j$, $\{j\}$ and $X_1\setminus T$.
		Obviously, we have $v_1(\hat{X}_1) \geq \frac{1}{3}\cdot v_1(X_1)$.
	\end{proof}
	
	By the above claim, we can now construct a new allocation $\mathbf{X}'$ with $X'_i = X_i$ for all $i\geq 3$, $X'_1 = T_l \cup \hat{X}_1$, $X'_2 = X_2 \setminus T_l$, and $X'_0 = X_0 \cup (X_1\setminus \hat{X}_1)$. 
	This allocation is feasible since
	\begin{equation*}
	c(X'_1) = c(T_l) + c(\hat{X}_1) \leq \frac{B_1}{2} + \frac{B_1}{2} = B_1.
	\end{equation*}	
	By construction, we have that 
	\begin{align*}
	v_1(X'_1) = &v_1(T_l)+v_1(\hat{X}_1) 
	>  \frac{5}{3}\cdot v_1(X_1)+ \frac{1}{3}\cdot v_1(X_1) = 2\cdot v_1(X_1),
	\end{align*}
	and $v_2(X'_2) \ge \frac{1}{2}\cdot v_2(X_2)$. Therefore
	\begin{equation*}
	\prod_{i \in N} v_i(X'_i) > \prod_{i \in N} v_i(X_i),
	\end{equation*}
	which is a contradiction.
\end{proof}

As we can see from the analysis, the improvement in the approximation ratio comes mainly from the fact that when reallocating items in $T_l$ from $X_2$ to $X_1$, we are able to keep a constant fraction of the items in $X_1$ (as in Claim~\ref{claim:keep-1/3}).

It is natural to expect that if we are able to obtain an even finer dissection of $X_1$, we should be able to improve the approximation ratio even further.
Indeed, this is the case.
We show that when all items are very small compared with the budgets, the approximation ratio approaches $1/2$.
While the idea is clear, as we will show in the following section, to accomplish the analysis is a highly non-trivial task.

\subsection{Large Budget Case}

Let $\kappa = \min_{i\in N,j\in M} \frac{B_i}{c_j}$.
Without loss of generality, assume that $\kappa$ is an integer; otherwise, we round $\kappa$ down to the nearest integer.  
In other words, it takes at least $\kappa$ items to exhaust the budget of any agent.

\begin{theorem}\label{thm:large-capacity}
	The Max-NSW allocation $\mathbf{X^*}$ is $\left( \frac{1}{2}-\frac{5}{\kappa^{1/4}} \right)$-EF1, where $\kappa = \min_{i\in N,j\in M}\frac{B_i}{c_j}$.
\end{theorem}

Observe that when $\kappa \rightarrow \infty$ then the approximation ratio (with respect to EF1) approaches $1/2$.

Before we present the proof of Theorem~\ref{thm:large-capacity}, we first show that the $1/2$ approximation ratio is tight.
In particular, we give an example with arbitrarily large $\kappa$, for which a Max-NSW allocation is $(\frac{1}{2} + O(\frac{1}{\kappa}))$-EF1.

\begin{example}
	Let there be two agents $\{1,2\}$ with $B_1 = B_2 = \kappa$.
	Let there be $2\kappa$ items $M = M_1\cup M_2$ with cost $1$, where $M_1 = \{1,2,\ldots,\kappa\}$ and $M_2 = \{\kappa+1,\kappa+2,\ldots,2\kappa\}$.
	Let the valuations be $v_{1j} = 1$ for all $j\in M_1$; $v_{1j} = 2$ for all $j\in M_2$; 
	$v_{2j} = 0$ for all $j\in M_1$; $v_{2j} = 2$ for all $j\in M_2$ (see below).
	
	\begin{table}[htp]
		\begin{center}
			\begin{tabular}{ccccccc}
				\toprule
				$j$\quad~	&	1			&	$\ldots$		&	$\kappa$	&	$\kappa + 1$ & $\ldots$ & $2\kappa$  \\
				\midrule
				$v_{1j}$\quad~	&	1	& $\ldots$	& 1 & 2 & $\ldots$ & 2 \\
				$v_{2j}$\quad~	&	0	& $\ldots$	& 0 & 2 & $\ldots$ & 2 \\
				\bottomrule
			\end{tabular}
		\end{center}
	\end{table}
	
	Suppose that $x\leq \kappa$ items in $M_2$ are allocated to agent $2$ in a Max-NSW allocation, then it is optimal to allocate the remaining $\kappa-x$ items in $M_2$ together with $x$ items in $M_1$ to agent $1$.
	The NSW is thus given by
	\begin{equation*}
	v_1(X_1)\cdot v_2(X_2) = (2\kappa - x)\cdot 2x,
	\end{equation*}
	which is maximized when $x = \kappa$.
	That is, $X_1 = M_1$ and $X_2 = M_2$.
	Since for any $j\in X_2$,
	\begin{equation*}
	v_1(X_2 - j) = 2(\kappa-1) = \frac{2(\kappa-1)}{\kappa}\cdot v_1(X_1),
	\end{equation*}
	the Max-NSW allocation is $\frac{\kappa}{2\kappa-1}$-EF1, and the approximation ratio approaches $\frac{1}{2}$ when $\kappa\rightarrow \infty$.
\end{example}

\medskip

\subsection{Proof of Theorem~\ref{thm:large-capacity}}

Next we prove Theorem~\ref{thm:large-capacity}.
For ease of notation we let $k = \kappa^{1/4}$.
Note that it suffices to consider the case when $k > 20$, as otherwise the theorem follows from Theorem~\ref{thm:Nash-sw-EF1}.

\begin{definition}[Density]
	Let the density of item $j\in M$ under the valuation of agent $i$ as $\rho_{ij} = \frac{v_{ij}}{c_j}$.
\end{definition}

We first prove the following lemma, which will be useful in the analysis.

\begin{lemma}\label{lemma:items-to-remove}
	Given two sets of items $X$ and $Y$, both have cost at most $B$, and a valuation function $v$, there exists a subset of items $Z\subseteq X\cup Y$ with $c(Z)\leq B$ and
	\begin{equation*}
	v(Z)\geq \left(1-\frac{c(Y)}{B}-\frac{1}{k^4} \right)\cdot v(X)+v(Y).
	\end{equation*}
\end{lemma}
\begin{proof}
	We prove by constructing the subset $Z$ satisfying the claimed properties.
	In particular, we compute a subset $X' \subseteq X$ with $c(X') \leq B - c(Y)$, and
	\begin{equation*}
	v(X') \geq \left(1-\frac{c(Y)}{B}-\frac{1}{k^4} \right)\cdot v(X).
	\end{equation*}
	Then the lemma follows immediately by setting $Z = X'\cup Y$.
	
	To ease the analysis, we assume without loss of generality that $c(X) = B$.
	If $c(X) < B$ then imagine that we include into $X$ many items with $0$ value and arbitrarily small cost until $c(X) = B$, which does not increase the value of $X$.
	We construct $X'$ by repeatedly removing an item in $X$ with the smallest density (under valuation function $v$), until $c(X') \leq B - c(Y)$.
	Since each item has cost at most $\frac{B}{k^4}$ by definition of $k$, the total cost of items removed in this procedure is at most $c(Y) + \frac{B}{k^4}$.
	Since the items removed are among those with the smallest density, their average density is at most the overall density of items in $X$; the total value of items removed is at most
	\begin{equation*}
	\left( c(Y) + \frac{B}{k^4} \right)\cdot \frac{v(X)}{c(X)} = \left( \frac{c(Y)}{B} + \frac{1}{k^4} \right)\cdot v(X),
	\end{equation*}
	
	Thus we have
	\begin{align*}
	v(Z) =  v(X') + v(Y) 
	\geq  \left(1-\frac{c(Y)}{B}-\frac{1}{k^4} \right)\cdot v(X) + v(Y),
	\end{align*}
	as claimed.
\end{proof}

As before, we prove by contradiction, and assume the allocation is not $(\frac{1}{2}-\frac{5}{k} )$-EF1, say, between agents $1$ and $2$.
We show that this assumption leads to contradictions.

\smallskip

Let $T\subseteq X_2$ with $c(T)\leq B_1$ such that
\begin{align}
v_1(T - j^*) > &\left(\frac{1}{2} - \frac{5}{k} \right)^{-1}\cdot v_1(X_1)
=  \left( 2+\frac{20}{k-10} \right)\cdot v_1(X_1) \label{eqn:hat-T-in-agent-1},
\end{align}
where $j^* = \arg\max_{j\in T}v_{2j}$ is the item in $T$ with the maximum value under the valuation of agent $2$.

Let $\hat{T} = T-j^*$.
For each item $j\in \hat{T}$, we refer to $v_{2j}/v_2(\hat{T})$ as the \emph{contribution} of item $j$ to set $\hat{T}$ under valuation of agent $2$.
The total contribution of items in $\hat{T}$ is $1$.
Depending on the contributions of items, we partition $\hat{T}$ into a set $T_h$  of \emph{heavy items} and a set $T_l$ of \emph{light items}:
\begin{align*}
T_h  := \left\{ j\in \hat{T} : \frac{v_{2j}}{v_2(\hat{T})} \geq \frac{1}{k^3}  \right\}, \quad
\text{ and } \quad
T_l  := \hat{T} \setminus T_h.
\end{align*}

We first show that under the valuation of agent $1$, the contribution of heavy items must be small.

\begin{claim}
	For all $j\in T_h$, we have $\frac{v_{1j}}{v_1(\hat{T})} < \frac{1}{2}\cdot \frac{v_{2j}}{v_2(\hat{T})}$.
\end{claim}
\begin{proof}
	Suppose for the sake of contradiction that there exists $j\in T_h$ such that
	\begin{equation*}
	\frac{v_{1j}}{v_1(\hat{T})} \geq \frac{1}{2}\cdot \frac{v_{2j}}{v_2(\hat{T})}.
	\end{equation*}
	We show that there exists an allocation (by reallocating item $j$) that achieves a strictly larger NSW, which is a contradiction.
	Let $c = \frac{v_2(\hat{T})}{v_{2j}}$.
	By definition of heavy item, we have $c \leq k^3$.
	Consider the allocation obtained by moving item $j$ from $\hat{T_h}$ to $X_1$.
	Since $v_{2j}\leq v_{2j^*}$, we have
	\begin{align*}
	v_2(X'_2) = & \left( 1 - \frac{v_{2j}}{v_2(X_2)} \right)\cdot v_2(X_2)
	\geq  \left( 1 - \frac{v_{2j}}{v_2(\hat{T}) + v_{2j^*}} \right)\cdot v_2(X_2) 
	\geq \left( 1 - \frac{1}{c+1} \right)\cdot v_2(X_2).
	\end{align*}
	
	Note that including $j$ into $X_1$ may result in a violation of the budget constraint of agent $1$.
	To resolve this issue, we use Lemma~\ref{lemma:items-to-remove} with $B = B_1$, $v = v_1$, $X = X_1$ and $Y = \{j\}$, and let $X'_1 = Z$ (as specified in Lemma~\ref{lemma:items-to-remove}).
	Since item $j$ has cost $c_{j} \leq \frac{B_1}{k^4}$, we have
	\begin{equation*}
	v_1(X'_1) \geq \left( 1-\frac{2}{k^4} \right)\cdot v_1(X_1) + v_{1j}.
	\end{equation*}
	
	Recall that by assumption of the proof we have
	\begin{equation*}
	v_{1j} \geq \frac{1}{2}\cdot v_1(\hat{T})\cdot \frac{v_{2j}}{v_2(\hat{T})} = \frac{1}{2c}\cdot v_1(\hat{T}).
	\end{equation*}
	
	By inequality~\eqref{eqn:hat-T-in-agent-1}, we have
	\begin{align*}
	v_1(X'_1) \geq & \left( 1-\frac{2}{k^4} + \frac{1}{2c}\cdot \left( 2+\frac{20}{k-10} \right) \right)\cdot v_1(X_1) \\
	= & \left( 1 + \frac{1}{c} + \frac{10}{c(k-10)}-\frac{2}{k^4} \right)\cdot v_1(X_1)  \\
	> & \frac{c+1}{c}\cdot v_1(X_1),
	\end{align*}
	where the last inequality holds since $c\leq k^3$.
	
	Hence, we have $v_1(X'_1)\cdot v_2(X'_2) > v_1(X_1)\cdot v_2(X_2)$ and
	\begin{equation*}
	\prod_{i \in N} v_i(X'_i) > \prod_{i \in N} v_i(X_i),
	\end{equation*}
	which is a contradiction.
\end{proof}

Let $f = \frac{v_2(T_l)}{v_2(\hat{T})}$ be the total contribution of light items to $\hat{T}$.
By the above claim we have
\begin{equation*}
\frac{v_1(T_l)}{v_1(\hat{T})} = 1 - \frac{v_1(T_h)}{v_1(\hat{T})} > 1 - \frac{1}{2}\cdot \frac{v_2(T_h)}{v_2(\hat{T})} = \frac{1}{2} + \frac{f}{2}.
\end{equation*}
In other words, under the valuation of agent $1$, the light items have a larger total contribution to $\hat{T}$.

Next, we prove the following lemma.
Roughly speaking, we can partition the light items into $k$ sets with about the same cost and value.

\begin{lemma}\label{lemma:even-fractional-partition}
	We can partition $T_l$ into $k$ sets such that each set $Y$ satisfies
	\begin{equation*}
	v_2(Y) \leq \left( \frac{f}{k}+\frac{4}{k^3} \right)\cdot v_2(\hat{T}), \;
	c(Y) \leq \left( \frac{1}{k}+\frac{4}{k^4} \right)\cdot B_1.
	\end{equation*}
\end{lemma}

For continuity of presentation, we defer the proof of the above lemma to the end of this section.

\smallskip

We partition $T_l$ into $k$ sets as specified in Lemma~\ref{lemma:even-fractional-partition}, then pick the one most valuable to agent $1$ out of the $k$ sets, and reassign it to agent $1$.
Let $Y$ be the set chosen.
Note that $Y$ satisfies
\begin{align*}
v_1(Y) \geq \frac{1}{k}\cdot v_1(T_l) \geq & \frac{1}{k}\left(\frac{1}{2}+\frac{f}{2}\right)\cdot v_1(\hat{T}) 
\geq  \frac{1}{k}\left(\frac{1}{2}+\frac{f}{2}\right)\left( 2+\frac{20}{k-10} \right)\cdot v_1(X_1).
\end{align*}

After losing set $Y$, the remaining value of agent $2$ is
\begin{align}
v_2(X'_2) =  v_2(X_2) - v_2(Y)
\geq  \left( 1-\frac{f}{k}-\frac{4}{k^3} \right) \cdot v_2(X_2), \label{eqn:bound-of-X2}
\end{align}
where the inequality holds by Lemma~\ref{lemma:even-fractional-partition}.
After acquiring set $Y$, agent $1$ has the set of items $X'_1 = X_1\cup Y$, which might have a larger cost than $B_1$.
Similarly, we use Lemma~\ref{lemma:items-to-remove}, with $B = B_1$, $v=v_1$ and $X = X_1$.
Let $X'_1 = Z$ be the set specified in Lemma~\ref{lemma:even-fractional-partition}.

As a result, the bundle agent $1$ receives has value
\begin{align}
 v_1(X'_1) &\geq \left( 1 - (\frac{1}{k}+\frac{5}{k^4}) + \frac{1}{k}(\frac{1}{2}+\frac{f}{2})( 2+\frac{20}{k-10} ) \right)\cdot v_1(X_1) \nonumber \\
&=  \left( 1+\frac{f}{k} + \frac{10(1+f)}{k(k-10)} - \frac{5}{k^4} \right)\cdot v_1(X_1) >  \left( 1+\frac{f}{k} + \frac{9}{k^2} \right)\cdot v_1(X_1),\label{eqn:bound-of-X1}
\end{align}
where in the last inequality we use $f\geq 0$ and $k \geq 20$.

Combining inequalities~\eqref{eqn:bound-of-X1} and~\eqref{eqn:bound-of-X2}, we have
\begin{align*}
 \frac{v_1(X'_1) v_2(X'_2)}{v_1(X_1) v_2(X_2)} 
&>  \left( 1-\frac{f}{k}-\frac{4}{k^3} \right)\cdot \left( 1+\frac{f}{k} + \frac{9}{k^2} \right)  \\
&=  \left( 1-(\frac{f}{k}+\frac{4}{k^3})^2 + (\frac{9}{k^2} -\frac{4}{k^3}) ( 1-\frac{f}{k}-\frac{4}{k^3} ) \right) \\
&\geq \left( 1-(\frac{f}{k}+\frac{4}{k^3})^2 + \frac{5}{k^2} \right) >  1.
\end{align*}

It follows that
\begin{equation*}
\prod_{i \in N} v_i(X'_i) > \prod_{i \in N} v_i(X_i),
\end{equation*}
which is a contradiction.

\subsection{Proof of Lemma~\ref{lemma:even-fractional-partition}}


\begin{proofof}{Lemma~\ref{lemma:even-fractional-partition}}
	Let $S = c(T_l)$ be the total cost of items in $T_l$.
	Recall that $S\leq B_1$.
	We first show that we can partition the set $T_l$ \emph{fractionally} into $k$ sets $Y_1,Y_2,\ldots,Y_k$ with equal cost and value.
	That is, for all $i\in[k]$, we have
	\begin{equation}
	c(Y_i) = \frac{1}{k}\cdot S \text{ and } v_2(Y_i) = \frac{1}{k}\cdot v_2(T_l). \label{eqn:fractional-partition}
	\end{equation}
	
	In general, obtaining such fractional partitioning is easy because we can simply assign each $Y_i$ a $\frac{1}{k}$ fraction of every item in $T_l$.
	However, this would result in too many fractional items.
	The key to our analysis is to show that we can fractionally and evenly partition $T_l$ such that each $Y_i$ contains at most $4$ fractional items.
	Given such a fractional partitioning $(Y_1,Y_2,\ldots,Y_k)$ of $T_l$, we round the fractional items arbitrarily, e.g., assign the fractional item to the set $Y_i$ containing it with the smallest $i$. 
	Hence we obtain an integral partition $(Y'_1,Y'_2,\ldots,Y'_k)$ of $T_l$.
	Moreover, since each set $Y_i$ contains at most $4$ fractional items and each (integral) item has cost at most $\frac{B_1}{k^4}$, we have
	\begin{align*}
	c(Y'_i) \leq & c(Y_i) + 4\cdot \frac{B_1}{k^4} 
	= \frac{1}{k}\cdot S + \frac{4}{k^4}\cdot B_1 \leq 
	\left( \frac{1}{k}+\frac{4}{k^4} \right)\cdot B_1 .
	\end{align*}
	By definition of light items, each $j\in T_l$ has value $v_{2j} < \frac{1}{k^3}\cdot v_2(\hat{T})$.
	Recall that $v_2(T_l) = f\cdot v_2(\hat{T})$, we have
	\begin{align*}
	 v_2(Y'_i) \leq & v_2(Y_i) + \frac{4}{k^3}\cdot v_2(\hat{T}) 
	=  \frac{1}{k}\cdot v_2(T_l) + \frac{4}{k^3}\cdot v_2(\hat{T})
	= \left( \frac{f}{k}+\frac{4}{k^4} \right)\cdot v_2(\hat{T}) .
	\end{align*}
	
	Hence both conditions are satisfied.		
	It remains to compute the sets $Y_1,Y_2,\ldots,Y_k$.
	
	\smallskip
	
	We first introduce some notation for fractional sets.
	Each fractional item set $F$ can be represented by a vector $\mathbf{y} \in [0,1]^{|T_l|}$, where $\mathbf{y} = (y_1,\ldots,y_{|T_l|})$ and $y_j\in [0,1]$ represents the fraction of item $j\in T_l$ that is included in $F$.
	The cost of $F$ is denoted by $c(F) = \sum_{j\in T_l} y_j$.
	
	Given $T_l$ and $b\leq S = c(T_l)$, a fractional item set $T_l[0,b]$ is defined as follows.
	We initialize $T_l[0,b]$ to an empty set, and then repeatedly include into this set a densest item from $T\setminus T_l[0,b]$ (under the valuation of agent $2$), until the cost of $T_l[0,b]$ will exceed $b$ if we include the next such item, say item $j$, into $T_l[0,b]$; we then include into $T_l[0,b]$ only a fraction of $j$ that will make $c(T_l[0,b]) = b$ (see Figure~\ref{fig:T_l[0,b]}).
	
	\begin{figure}[h]
		\centering
		\vspace*{-5pt}
		\includegraphics*[width=0.4\textwidth]{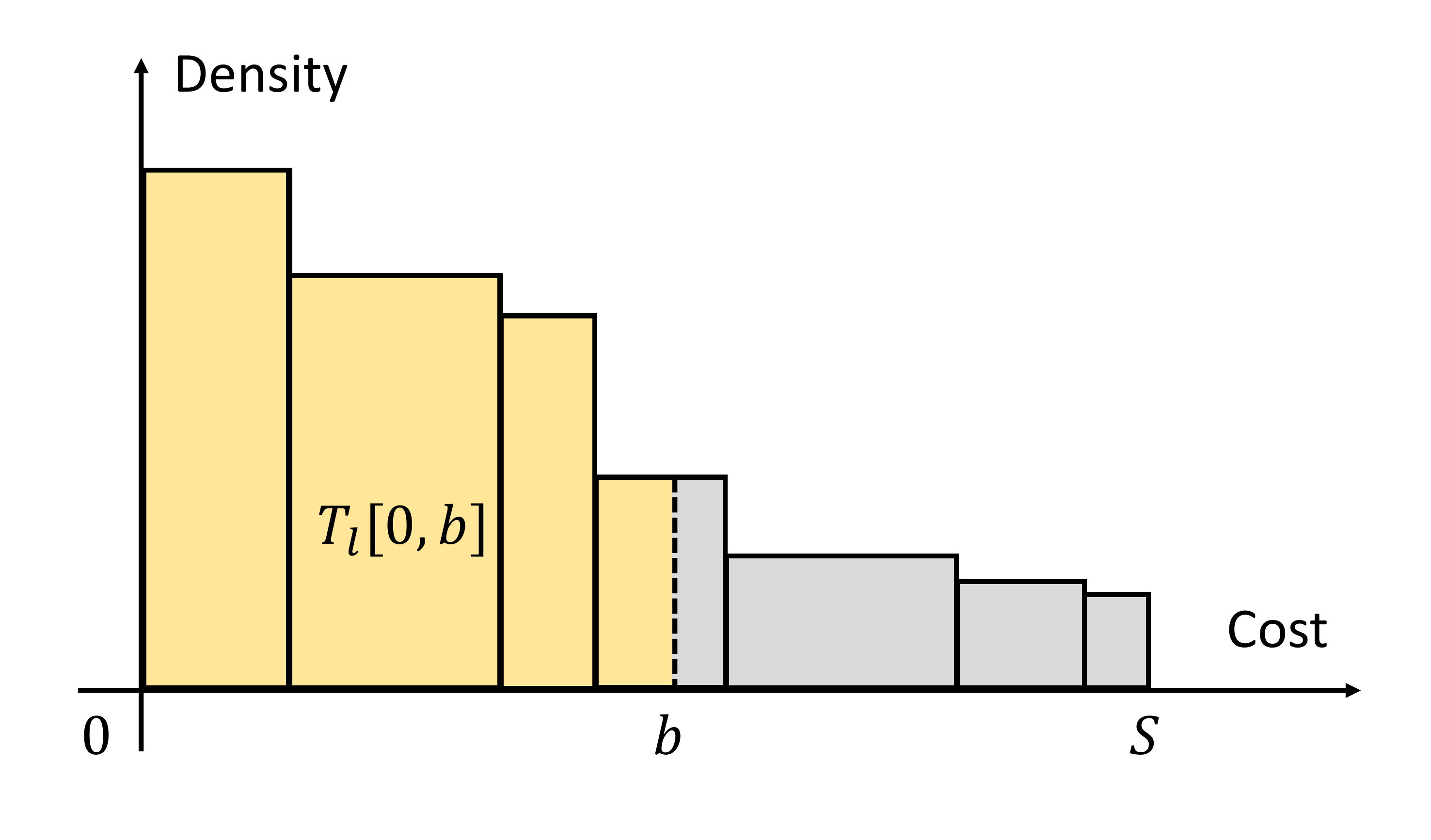}
		\vspace*{-15pt}
		\caption{Illustrating example of $T_l[0,b]$, where each item is represented as a rectangle with the height being its density and width being its cost. The items are listed in descending order of their densities.}
		\label{fig:T_l[0,b]}
	\end{figure}
	
	Equivalently, we describe the fractional item set $T_l[0,b]$ as a vector $\mathbf{y}\in [0,1]^{T_l}$, where for each $j\in T_l$,
	\begin{equation*}
	y_j = \min \left\{ 1,\ \frac{1}{c_j}\cdot \max\Big\{ 0, b-\sum_{j':\rho_{2j'}>\rho_{2j}} c_{j'} \Big\} \right\}.
	\end{equation*}
	
	Similarly, for each $0\leq a\leq b\leq S$ we define $T_l(a,b]$ to be the fractional item set obtained by letting
	\begin{equation*}
	T_l(a,b] := T_l[0,b] \setminus T_l[0,a].
	\end{equation*}
	
	\smallskip
	
	Next we construct the sets $Y_1,Y_2,\ldots,Y_k$.
	To construct $Y_1$, we use a parameter $b_1\in [0,\frac{S}{k}]$ and let
	\begin{equation}
	\label{eq:def-Y1}
	Y_1 = T_l[0,b_1] \cup T_l\left( \frac{k-1}{k}\cdot S + b_1,S \right].
	\end{equation}
	Note that for any $b_1 \in [0,\frac{S}{k}]$, we have $c(Y_1) = \frac{S}{k}$, whereas the value $v_2(Y_1)$ increases (weakly) with $b_1$.
	For $b_1 = 0$, $Y_1 = T_l\left( \frac{k-1}{k}\cdot S,S \right]$ contains the fractional items in $T_l$ with the smallest densities; for $b_1 = \frac{S}{k}$, $Y_1 = T_l\left[0, \frac{S}{k} \right]$ contains the fractional items in $T_l$ with the largest densities.
	Hence, there exists a value $b_1 \in [0,\frac{S}{k}]$ which makes $v_2(Y_1) = \frac{1}{k}\cdot v_2(T_l)$.
	We fix $b_1$ to this value and let $Y_1$ be the corresponding set defined by Equation~\eqref{eq:def-Y1}. Next we construct $Y_2$.
	
	\smallskip
	
	Note that after fixing $Y_1$, the remaining fractional item set is $T_l\left( b_1, \frac{k-1}{k}\cdot S+b_1 \right]$, which has cost $\frac{k-1}{k}\cdot S$ and value
	\begin{equation*}
	v_2(T_l) - v_2(Y) = \frac{k-1}{k}\cdot v_2(T_l).
	\end{equation*}
	In other words, the remaining items have average density exactly the same as $\rho_2(T_l) = \frac{v_2(T_l)}{c(T_l)}$.
	Thus we can construct $Y_2$ the same way we constructed $Y_1$:
	for $b_2 \in [b_1, \frac{S}{k} + b_1]$, let
	\begin{equation*}
	Y_2 = T_l( b_1, b_2 ] \cup T_l\left( \frac{k-2}{k}\cdot S+b_2, \frac{k-1}{k}\cdot S+b_1 \right].
	\end{equation*}
	Similarly, for any $b_2$ we have $c(Y_2) = \frac{S}{k}$, and the value $v_2(Y_2)$ increases with $b_2$. We can find a value for $b_2$ that makes $v_2(Y_2) = \frac{1}{k}\cdot v_2(T_l)$ and let $Y_2$ be the corresponding set defined by Equation~\eqref{eq:def-Y1}.
	
	\begin{figure}[htb]
		\centering
		\vspace*{-5pt}
		\includegraphics*[width=0.4\textwidth]{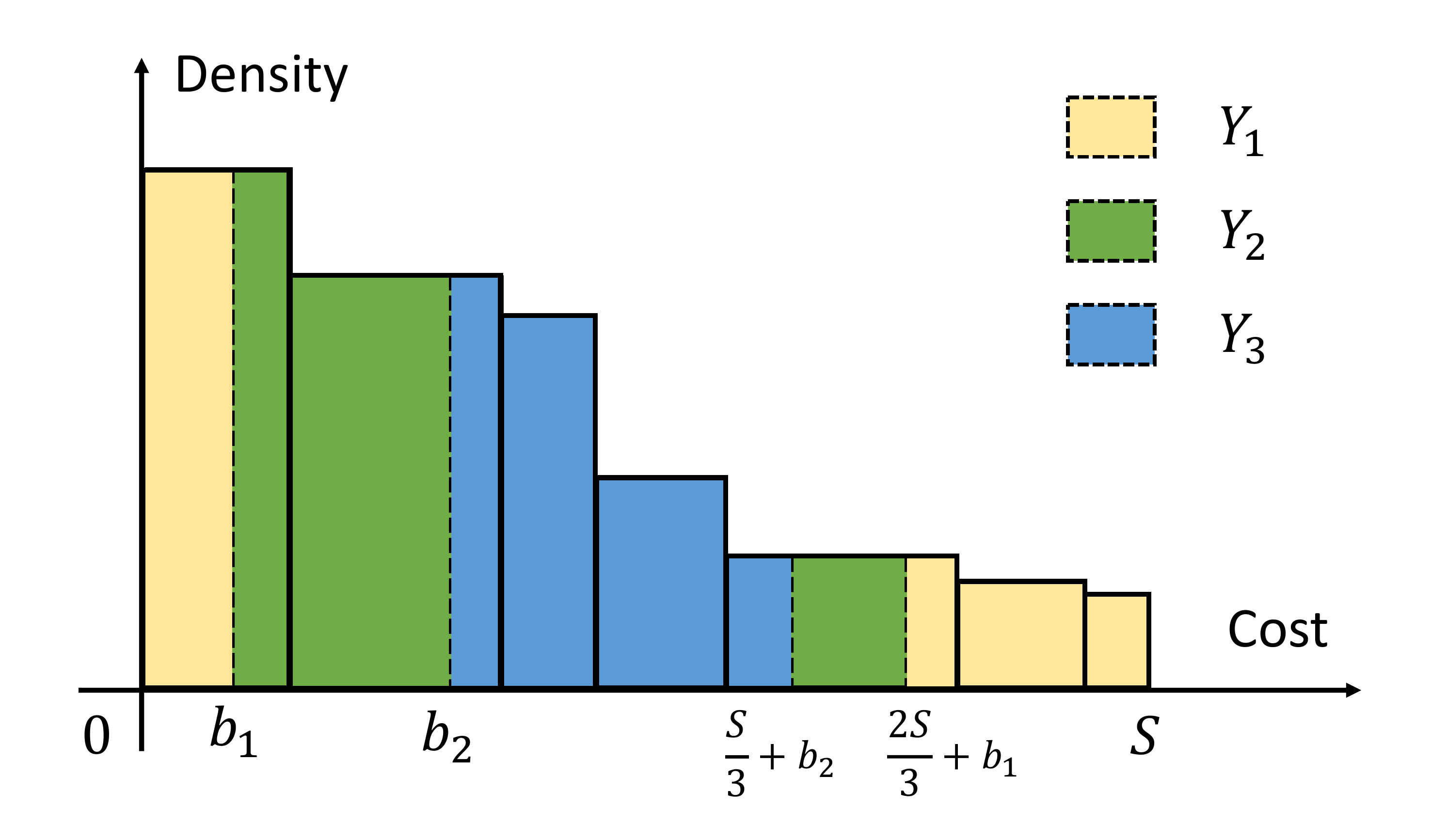}
		\vspace*{-15pt}
		\caption{Partitioning $T_l$ into disjoint fractional item sets: an example with $k=3$.}
		\label{fig:Y}
	\end{figure}

	By repeating the above procedure, we can construct the subsequent sets $Y_3,\dots, Y_k$ and obtain $k$ fractional sets with the equal cost and value (see Figure~\ref{fig:Y}).
	Note that each fractional set $Y_i$ contains at most four fractional items, one for each interval boundary in Equation~\eqref{eq:def-Y1}.
	Hence, the sets $Y_1,\ldots,Y_k$ satisfy the claimed properties. 
\end{proofof}

\section{Discussion and Future Directions}

We showed that in the presence of budget constraints a Max-NSW allocation may not be EF1 but achieves a constant approximation of EF1.
The tight approximation ratio is $1/4$ and in the case of large budgets, the ratio improves with the budget-cost ratio and converges to $1/2$ when this ratio goes to infinity.
Our results complement previous work that focused on the setting without budget constraint.

There are several directions for future work.
First, our results actually imply that in a budget-feasible setting, a $1/4$-EF1 allocation always exists and is compatible with Pareto optimality, but it is not known whether an exact EF1 allocation always exists or not and this appears to be a non-trivial problem.
Second, the computation of maximum NSW in our setting is an interesting question which we did not study in this paper. 
Indeed, an APX-hardness for the complexity of this problem can be readily established since the setting without budget can be seen as a special case of our setting (where every agent has budget $c(M)$).
Thus, an natural task is to design efficient approximation algorithms.  
Theorem~\ref{thm:Nash-sw-EF1} can be extended to show that an $\alpha$-approximation of the Max-NSW is also $\frac{\alpha}{4}$-EF1.
Hence an approximation algorithm for Max-NSW would also compute an approximate EF1 allocation.   
However, in the large budget case, an $\alpha$-approximation of the Max-NSW is not necessarily $\alpha\cdot \left( \frac{1}{2}-\frac{5}{\kappa^{1/4}} \right)$-EF1 due to the following example.
\begin{example}
	Fix an arbitrarily large integer $\kappa$.
	Suppose there are two agents with capacity $B_1 = B_2 = \kappa$, and $2\kappa$ items of size $1$. Let the first $\kappa$ items $M_1 = \{1,2,\ldots,\kappa\}$ have value $1$ and the remaining $\kappa$ items $M_2 = \{\kappa+1,\kappa+2,\ldots,2\kappa\}$ have value $0.2$.
	Obviously the optimal Nash social welfare is $0.36\kappa^2$.
	Now consider the allocation $\mathbf{X} = (M_1,M_2)$.
	The Nash social welfare is $0.2\kappa^2$, which implies an approximation ratio $\alpha > 0.5$.
	However, the allocation is not $c$-EF1 for any $c > \frac{1}{5}+\frac{1}{5(\kappa-1)}$.
\end{example}

Finally, while we focus on EF in this paper, it would also be interesting to consider other fairness notions such as proportionality and maximin share fairness and study whether a Max-NSW allocation provides any fairness guarantee under these notions.

\section*{Acknowledgement}

The authors thank Edith Elkind, Georgios Birmpas, Warut Suksompong, 
and Alexandros Voudouris for helpful discussions at the early stage of this work.

\bibliographystyle{abbrv}
\bibliography{nashsw}

\newpage

\appendix

\section{Sub-additive Evaluation Functions}\label{appendix:sub-additive}

While we assume that the valuation functions are additive throughout this paper, 
Theorem \ref{thm:Nash-sw-EF1} can be extended to sub-additive valuation functions as well.
A valuation function $v : 2^M \rightarrow \mathbb{R}^+\cup\{0\}$ is sub-additive if for any disjoint sets $S,T\subseteq M$, we have
\begin{equation*}
v(S) + v(T) \geq v(S\cup T).
\end{equation*}

\begin{theorem}[Sub-additive Valuations]\label{thm:Nash-sw-EF1-sub-additive}
	When all agents have sub-additive valuation functions, a Max-NSW allocation $\mathbf{X^*}$ is $\frac{1}{4}$-EF1 and PO.
\end{theorem}
\begin{proof}
	The proof is almost identical to that of Theorem~\ref{thm:Nash-sw-EF1}, and the only difference is here we provide a more general way to partition $T$ into $T_1$, $T_2$ and $\{t\}$ such that for some $i\in\{1,2\}$,
	\begin{align}\label{eq:sub-additive}
	    v_2(T_i) \le v_2(T_{3-i}) \le v_2(T_i+ t).
	\end{align}
    
    Let $T'_1$ and $T'_2$ be the two sets returned by Algorithm \ref{alg:sub-additive} on inputs $T$ and $v_2$. Then we have $v_2(T'_1) > v_2(T'_2)$.
    Moreover, for any $t \in T'_1$, it holds that $v_2(T'_1-t) \le v_2(T'_2+t)$.
    Let $t$ be an arbitrary item in $T'_1$; let $T_1 = T'_1 - t$, and $T'_2 = T_2$. 
    Then $T_1$, $T_2$ and $\{t\}$ satisfy \eqref{eq:sub-additive}.
\end{proof}

\begin{algorithm}[htbp]
	\caption{Partioning $T$.\label{alg:sub-additive}}
	{\bf Input:} A set of items $T$ and a sub-additive valuation $v$.
	
	Initialize: $T_1 = T$ and $T_2 = \emptyset$.\\
	\While{there exists $e \in T_1$ such that $v(T_1-e) > v(T_2+ e)$}
	{
		$T_1 = T_1-e$ and $T_2 = T_2+ e$. \\
	}
	{\bf Output:} A partition of $T = (T_1, T_2)$.
\end{algorithm}

\end{document}